\documentclass[a4paper,reqno]{amsart}
\usepackage[text={5.2in,8in},centering]{geometry}
\usepackage{amsrefs}
\usepackage{amssymb}
\usepackage{mathrsfs}

\DeclareMathAlphabet{\mathpzc}{OT1}{pzc}{m}{it}

\usepackage{enumerate}
\usepackage{graphicx}
\numberwithin{equation}{section}
\theoremstyle{plain}
\newtheorem{theorem}{Theorem}[section]

\newtheorem{lemma}{Lemma}[section]
\newtheorem{corollary}[lemma]{Corollary}
\theoremstyle{remark}
\newtheorem{remark}{Remark}[section]

\newtheorem*{quest*}{Question}
\newtheorem*{remark*}{Remark}
\theoremstyle{remark}
\theoremstyle{definition}
\newtheorem{definition}{Definition}[section]
\newtheorem*{definition*}{Definition}
\newtheorem*{notation*}{Notation}
\newtheorem*{notations*}{Notations}
\providecommand{\B}{\mathbf}

\providecommand{\C}{\mathcal}

\providecommand{\D}{\mathbb}

\providecommand{\R}{\mathrm}
\newcommand{\ee}{\mathrm{e}}
\newcommand{\eul}{\mathrm{e}}
\newcommand{\ii}{\mathrm{i}}
\def\ball{\mathrm{B}}

\DeclareMathOperator{\dist}{dist}

\DeclareMathOperator*{\essup}{ess\,sup}
\DeclareMathOperator{\expect}{\mathbb{E}}

\DeclareMathOperator{\one}{\mathbf{1}}

\def\tabhigh#1{{\buildrel _{} \over {#1}}}

\def\loc{$m\mathcal{-L}${\textit{oc}}}
\def\nloc{$m\mathcal{-N}\mathcal{L}${\textit{oc}}}

\def\Iloc{$(m,I)\mathcal{-L}${\textit{oc}}}
\def\Inloc{$(m,I)\mathcal{-N}\mathcal{L}${\textit{oc}}}

\def\mathloc{m\mathcal{-L}{\textit{oc}}}

\def\mathnloc{m\mathcal{-N}\mathcal{L}{\textit{oc}}}

\def\mathntun{m\rm{-NT}}
\def\mathtun{m\rm{-T}}

\def\lam{{\lambda}}
\def\Lam{{\Lambda}}

\def\BN{\B{N}}
\def\BD{\B{D}}
\def\BP{\B{P}}

\def\Const{{\rm{Const}}}

\def\DC{\D{C}}

\def\DP{\D{P}}
\def\DR{\D{R}}
\def\DZ{\D{Z}}
\def\DN{\D{N}}

\def\CE{\C{E}}
\def\CF{\C{F}}

\def\CL{\C{L}}
\def\CM{\C{M}}
\def\CN{\C{N}}

\def\CR{\C{R}}
\def\CS{\C{S}}
\def\CT{\C{T}}

\def\csB{\mathscr{B}}

\def\rd{{\R{d}}}
\def\rtd{{\widetilde{\R{d}}}}

\def\hx{\hat{x}}

\def\mubxy{{\mu^{x,y}_{\ball,\om}}}
\def\mukxy{{\mu^{x,y}_{\ball_{L_k}(0),\om}}}
\def\muxy{{\mu^{x,y}_{\om}}}

\def\be{\begin{equation}}
\def\ee{\end{equation}}
\def\ba{\begin{array}{l}}
\def\ea{\end{array}}
\def\bal{\begin{aligned}}
\def\eal{\end{aligned}}

\def\Cmix{\B{Mix(\delta)}}

\def\fF{\mathfrak{F}}

\def\om{{\omega}}
\def\Om{{\Omega}}

\def\pr#1{\D{P}\left\{\,#1\,\right\}}
\def\esm#1{\D{E}\left[\, #1\, \right]}
\def\pt{\partial}
\def\half{\frac{1}{2}}

\def\quart{\frac{1}{4}}

\def\truc#1#2#3{\smash{\mathop{\,\, #1 \,\, }\limits^{#2}_{#3}}}

\def\tto#1{\smash{\mathop{\,\,\,\, \longrightarrow \,\,\,\, }\limits_{#1}}}

\def\myset#1{{\left\{\,#1\,\right\}}}

\def\ketbra#1#2{{ | {#1} \rangle \langle {#2}| }}
\def\bra#1{{ \langle {#1}| }}
\def\ket#1{{ | {#1} \rangle }}

\def\mymax#1{{ \truc{\max} {} {#1}}}

\def\diy{\displaystyle}

\def\Wone{{\bf (W1)}}
\def\Wtwo{{\bf (W2)}}
\def\Wthree{{\bf (W3)}}

\def\mynobreak{\penalty 100 \hfilneg \hbox}

\def\tabhigh#1{{\buildrel _{} \over {#1}}}

\begin{document}
%---------------------------------------------------------------------%

\title[Direct Scaling Analysis. I.  Single-particle models]
{Direct Scaling Analysis \\of  localization in disordered systems. \\
  I. Single-particle models }
%---------------------------------------------------------------------%

\author[V. Chulaevsky]{Victor Chulaevsky}

%---------------------------------------------------------------------%

\address{D\'{e}partement de Math\'{e}matiques\\
Universit\'{e} de Reims, Moulin de la Housse, B.P. 1039\\
51687 Reims Cedex 2, France\\
E-mail: victor.tchoulaevski@univ-reims.fr}

%\date{\today}
\date{}
%---------------------------------------------------------------------%
\begin{abstract}
We propose a simplified version of the Multi-Scale Analysis  of tight-binding Anderson models with strongly mixing random potentials which leads directly to uniform exponential bounds on decay of eigenfunctions in  arbitrarily large finite subsets of a lattice. Naturally, these bounds imply also dynamical localization and exponential decay of eigenfunctions on the entire lattice.
\end{abstract}

\maketitle

%---------------------------------------------------------------------%
\section{Introduction. } \label{sec:intro}
%---------------------------------------------------------------------%

In this paper, we study spectral properties of random lattice Schr\"{o}dinger operators (LSO)  in the framework of the Multi-Scale Analysis  (MSA) developed in pioneering works \cites{FS83,FMSS85,DK89,DK91}. Traditionally, the MSA is applied first to the resolvents
\mynobreak{$G(E) = (H-E)^{-1}$} in finite subsets of growing size, usually balls $\ball_{L_k}(u)$
of radius $L_k = (L_0)^{\alpha^k}$, $\alpha>1$, $k \ge 0$. The bounds on the kernels $G_{\ball_{L_k}(u)}(x,y; E)$, obtained by scale induction,
are used then to derive an exponential decay of the eigenfunctions in $\DZ^d$. They can be re-used again to obtain the dynamical localization bounds.

In contrast with the MSA, a more recent method developed by Aizenman and Molchanov \cite{AM93}
and called the Fractional-Moment Method (FMM), when applicable, leads directly to the proof of dynamical localization; the latter, as it is well-known, implies spectral localization.
Despite striking differences between the MSA and the FMM, both approaches to the localization  share a certain similarity: the analysis of decay properties  of the resolvents in finite volumes \emph{precedes} the study of the eigenfunction correlators. In the method described below, the analysis of eigenfunctions is the dominant component.
As a result, the simplified Scaling Analysis allows to prove  localization bounds in finite volumes in the course of the scale induction.
We would like to emphasize that it is close in spirit to an elementary method used earlier by Spencer \cite{Spe88} to prove exponential decay of Green functions at a fixed energy.
The idea of a direct analysis of eigenfunctions is  inspired by a work of Sinai \cite{Si87} where it was implemented in a different context (quasi-periodic 1D lattice Schr\"{o}dinger operators) and in a different way (with the help of a KAM-type scale induction).

In a forthcoming work \cite{C11} our method is adapted to the multi-particle setting.

The structure of this manuscript is as follows:

\begin{enumerate}
  \item Basic notions and notations are introduced in Section \ref{sec:basic.notions}.
  \item In Section \ref{ssec:radial.descent}, we give a streamlined formulation of the main analytic tool of the MSA, allowing to establish exponential decay of eigenfunctions and Green functions in finite balls in  absence of multiple "resonances". A more traditional version of this technique, going back to \cite{FMSS85} and \cite{DK89}, requires a number of additional, rather tedious (albeit elementary) geometrical arguments. (A reader familiar with \cite{DK89} may want to skip Section \ref{ssec:radial.descent}). The two central notions of the new approach, allowing to simplify the MSA, are introduced in Section \ref{ssec:loc.tun} (cf. Definitions \ref{def:loc} and  \ref{def:tun.IID}).
  \item Section \ref{sec:SimpleScaleInd} describes the simplified scale induction in a particular case of an IID random potential of large amplitude ("strong disorder"). Uniform bounds on eigenfunction correlators, obtained in the course of the scale induction, imply dynamical localization at any finite scale, with decay rate faster than polynomial, while a more traditional approach gives rise to a power-law decay.\footnote{A faster (sub-exponential) rate can be achieved by a more sophisticated "bootstrap" procedure; cf. e.g., \cite{GK01} }
  \item In Section \ref{sec:adapt.low.E}, we consider weakly disordered media where localization can be established only in a specific energy band (at "extreme energies").
  \item In Section \ref{sec:adapt.mixing}, we adapt our method to lattice models with strongly mixing  random potentials.  Recall that the first results in this direction (in the multi-dimensional context) were obtained by von Dreifus and Klein \cite{DK91}.

\end{enumerate}

The proofs of all statements not given in the main text can be found in Appendix.

\section{Basic definitions and assumptions. } \label{sec:basic.notions}

\subsection{Assumptions}

We  consider the lattice Anderson Hamiltonians of the  form
\be\label{eq:H}
H(\omega) =  -\Delta + gV(x;\omega),
\ee
where $V:\DZ^d\times\Omega\to\DR$ is a  random field relative to some probability space
$(\Omega, \CF, \DP)$, and $\Delta$ is the nearest-neighbor lattice Laplacian.

For the sake of clarity, our method is presented first in a simpler situation where the random field $V$ is IID, with common marginal probability distribution function (PDF)
$F_{V}(t) = \pr{ V(0;\omega) \leq t}$. In this case, we assume:
\begin{enumerate}[\bf(W1)]
  \item The marginal PDF $F_V$ is uniformly H\"{o}lder-continuous: for some $b>0$,
\be\label{eq:V}
\forall\, t\in\DR, \, \forall\,\epsilon\in[0,1] \qquad  F_V(t+\epsilon) - F_V(t) \leq \Const \, \epsilon^b.
\ee
\end{enumerate}

H\"{o}lder continuity can be relaxed to log-H\"{o}lder continuity:
$F_V(t+\epsilon) - F_V(t) \leq \Const \ln^{-A} \epsilon^{-1}$ with sufficiently large $A>0$. This is sufficient for the proof of a polynomial decay of eigenfunction (EF) correlators and for dynamical localization at a polynomial rate. In order to prove decay of EF correlators faster than polynomial, one needs a regularity condition for $F_V$ slightly stronger that log-H\"{o}lder continuity; this is why we make the assumption \Wone.

In  Section \ref{sec:adapt.mixing}, we will consider a more general case of a correlated, but strongly mixing random potential.
Let $F_{V,x}(t) = \pr{ V(x;\omega) \leq t}$, $x\in\DZ^d$, be the marginal probability distribution functions (PDF)  of the random field $V$, and
$F_{V,x}(t\,|\, \fF_{\ne x}) := \pr{ V(x;\om) \le t \,|\, \fF_{\ne x}}$ the  conditional distribution functions (CDF) of the random field $V$ given the sigma-algebra
$\fF_{\ne x}$   generated by random variables $\{V(y;\om), y\ne x\}$. Our assumptions on correlated potentials are summarized as follows:

\begin{enumerate}[\bf(W2)]
%  \item The marginal PDF   are uniformly H\"{o}lder-continuous: for some $b>0$,
%\be\label{eq:V}
%\forall\, t\in\DR \;\; \forall\,\epsilon\in[0,1] \qquad
%F_{V,x}(t+\epsilon) - F_{V,x}(t) \leq \Const \, \epsilon^b.
%\ee
  \item The marginal CDFs are uniformly H\"{o}lder-continuous: for some $b>0$,
\be\label{eq:cond.continuity.V}
\essup \, \sup_{x\in\DZ^d}\, \sup_{a\in \DR}
\left[ F_{V,x}(a+\epsilon\,|\, \fF_{\ne x}) - F_{V,x}(a\,|\, \fF_{\ne x}) \right] \le \Const\,\epsilon^b.
\ee
\end{enumerate}
\begin{enumerate}[\bf(W3)]
  \item (Rosenblatt strong mixing condition.)
For any $L\ge 1$ and pair of balls $\ball_{L}(u)$, $\ball_{L}(v)$ with
$d(u,v)\ge L$,  for any events $\CE_{u,L}\in\fF_{\ball_{L}(u)}$, $\CE_{v,L}\in\fF_{\ball_{L}(v)}$
\be\label{eq:Cmix}
 \left| \pr{  \CE_{u,L} \cap \CE_{v,L} } - \pr{  \CE_{u,L} }  \pr{\CE_{v,L} } \right|
 \le C e^{- \ln^2 L}.
\ee

\end{enumerate}
This mixing condition can be relaxed to a power-law decay if only a polynomial decay of EF correlators is to be proven. (The polynomial decay of EF correlators is \emph{not} the decay of eigenfunctions which we always prove to be exponential.)

\subsection{Balls and boundary conditions}
\label{ssec:balls}

Unless otherwise specified, below we make use of the max-norm $\|x\| = \max_{1\leq j \leq d} |x_j|$ and the distance $\rd(\cdot,\cdot)$ induced by it. To describe the proximity in $\DZ^d$, we need the graph distance $\rtd(x,y)$ defined as the length of the shortest path from $x$ to $y$ formed by the lattice bonds.

Consider a lattice ball
$
\ball_\ell(u) = \myset{x:\, \| x-u\| \leq \ell}.
$
Note that if $\|\cdot\|$ is the max-norm, then the ball of radius $\ell$ is actually a (lattice) cube of side length $2\ell$ with sides parallel to coordinate hyperplanes. Introduce the following boundaries relative to $\ball_\ell$:
$$
\bal
\pt^- \ball_\ell(u) &= \myset{x\in \DZ^d:\, d(x,u) = \ell} \\
\pt^+ \ball_\ell(u) &= \myset{x\in \DZ^d\setminus\ball_\ell(u):\, \rtd(x, \ball_\ell(u)) =1} \\
\pt \ball_\ell(u)   &= \myset{ (x,y)\in\pt^- \ball_\ell(u)\times\pt^+ \ball_\ell(u):\; \rtd(x,y)=1}.
\eal
$$
Next, consider a pair of embedded balls $\ball_\ell(u)\subset\ball_L(w)$, $\ell<L$,  and the complementary set $\ball^c_\ell(u) := \ball_L(w)\setminus \ball_\ell(u)$. Introduce LSOs
$H_{\ball_L(w)}$, $H_{\ball_\ell(u)}$, $H_{\ball^c_\ell(u)}$ with Dirichlet boundary conditions
and the resolvents  $G_{\ball_L(w)}(E)$, $G_{\ball_\ell(u)}(E)$.
We will denote by $G(x,y;E)$ the matrix elements of resolvents (Green functions) in the standard basis of functions $\delta_x:y\mapsto\delta_{xy}$.
It follows from the second resolvent equation that, for any $x\in\ball_\ell(u)$, $y\in\ball^c_\ell(u)$ and any
$E\not\in \Sigma(H_{\ball_L(w)})\cup \Sigma(H_{\ball_\ell(u)})$
\be\label{eq:GRI.equal}
\bal
\left| G_{\ball_L}(x,y;E) \right|
& \le C_\ell   \max_{v \in  \pt^{{}^-}\ball_\ell(x)} |G_{\ball_\ell(x)}(x,v;E) |
\; \max_{v'\in  \pt^-\ball_\ell(v)} \, \left|G_{\ball_L}(v', y; E) \right|
\\
& \le \left(\, C_\ell   \max_{v: d(x,v)=\ell} |G_{\ball_\ell(x)}(x,v;E) | \,\right)
\; \max_{v: d(u,v)\le \ell+1} \, \left|G_{\ball_L}(v, y; E) \right|
\eal
\ee
(the so-called Geometric Resolvent Inequality)
where
$
C_\ell = |\pt  \ball_\ell(u)|  \le C(d) \ell^{d-1}.
%\ee
$
The latter bound, less precise than the former, is sufficient for the  MSA.
Similarly, for the eigenfunctions $\psi$ of operator $H_{\ball_L(w)}$ with eigenvalue
$E \not\in \Sigma(H_{\ball_\ell(u)})$ we have
\be\label{eq:GRI.EF}
| \psi(x) |
\le \left(\, C_\ell \,  \max_{v: d(x,v)=\ell} |G_{\ball_\ell(x)}(x,v;E) | \,\right)
\; \max_{y: d(u,y)\le \ell+1} \, | \psi(y) |,
\quad x\in\ball_\ell(u).
\ee
The same inequalities remain valid in the situation where $\ball_\ell(u)$ is not entirely a subset of $\ball_L(w)$. In such a case, however, one has to replace the ball $\ball_\ell(u)$  and its boundaries by their intersections with the set $\ball_L(w)$; see the details in Appendix.

%%%%%%%%%%%%%%%%%%%%%%%%%%%%%%%%%%%%%%%%%%%%%%%%%%%%%%%%%%%%%%%%%%%%%%%%%%%%%%%%%%%
% IID CASE

%%%%%%%%%%%%%%%%%%%%%%%%%%%%%%%%%%%%%%%%%%%%%%%%%%%%%%%%%%%%%%%%%%%%%%%%%%%%%%%%%%%
%%%%%%%%%%%%%%%%%%%%%%%%%%%%%%%%%%%%%%%%%%%%%%%%%%%%%%%%%%%%%%%%%%%%%%%%%%%%%%%%%%%
\section{Decay of Green functions and eigenfunctions }\label{sec:MSA.IID}

\subsection{"Radial descent" bounds for Green functions and eigenfunctions}
\label{ssec:radial.descent}

Now we fix the parameter $\beta\in(0,1)$; for our purposes, it is convenient  to take
$\beta > 1/4$, so that one can set the exponent (figuring in Lemma \ref{lem:WS1.IID}) $\beta'=1/4$.
\begin{definition}\label{def:CNR.IID}
Given a sample $V(\cdot;\om)$, a ball $\ball_L(u)$ is called
\begin{enumerate}[\;\;$\bullet$]
  \item $E$-non-resonant ($E$-NR) if $\|G_\ball(E;\om)\|\leq e^{+L^{\beta} }$,
  and $E$-resonant ($E$-R), otherwise;
  \item completely $E$-non-resonant ($E$-CNR) if it does not contain any $E$-R ball $\ball_\ell(u)$ with $\ell\ge L^{1/\alpha}$, and $E$-partially resonant ($E$-PR), otherwise.
\end{enumerate}
\end{definition}

%%%%%%%%%%%%%%%%%%%%%%%%%%%%%%%%%%%%%%%%%%%%%%%%%%%%%%%%%%%%%%%%%%%%%%%%%%%%%%%%%%%%
\begin{lemma}[Wegner-type bound]\label{lem:WS1.IID} Fix a finite interval $I\subset\DR$. There exist $L^*>0$ and $\beta' \in(0, \beta)$ such that
\begin{enumerate}[\rm (A)]
 \item for any $E\in I$, any $L\geq L^*$ and any ball $\ball_{L}(x)$
\begin{equation}\label{eq:W1.CNR.IID}
\pr{ \text{ $\ball_{L}(x)$ is not $E$-\rm{CNR} } }
\le e^{-L^{\beta'}};
\end{equation}
  \item for all $L\geq L^*$ and any pair of disjoint balls $\ball_{L}(x)$, $\ball_{L}(y)$
\begin{equation}\label{eq:W2.IID}
\pr{ \exists\, E\in I\, \text{ $\ball_{L}(x)$ and $\ball_{L}(y)$  are $E$-\rm{PR} } }
\le e^{-L^{\beta'}}.
\end{equation}
\end{enumerate}
\end{lemma}

In the case where the random potential is bounded one can take the interval $I$ containing the entire spectrum of the operator $H(\om)$.

\begin{definition}\label{def:SubH.IID}
Let be given two integers $L\ge \ell\ge 0$ and a  number $0 < q <1$.
Consider a finite connected subgraph $\ball\subset\DZ^d$. A function
$f:\, \ball\to\DR_+$ will be called $(\ell,q)$-subharmonic in
a ball $\ball_L(u)\subsetneq\ball$ if for any $\ball_\ell(x)\subset\ball_L(u)$ one has
\be\label{eq:def.l.q.subh.IID}
f(u) \leq q \;\;\mymax{y:\, \|y-u\|\le\ell+1} f(y).
\ee
\end{definition}

The relevance of this notion for the analysis of localization properties of eigenfunctions is explained by Lemma \ref{lem:cond.SubH.IID} below. Analytic (deterministic) bounds given by Lemma \ref{lem:SubH.IID} and Lemma \ref{lem:BiSubH.IID} can be considered as a convenient alternative to the well-known technique going back to \cite{FMSS85}, \cite{DK89}, also based on the GRI.

%%%%%%%%%%%%%%%%%%%%%%%%%%%%%%%%%%%%%%%%%%%%%%%%%%%%%%%%%%%%%%%%%%%%%%%%%%%%%%%%%%%%
\begin{lemma}\label{lem:SubH.IID}
Let $\ball$ be a finite connected subgraph of the lattice $\DZ^d$ and
$f: \ball\to\DR_+$ is an $(\ell,q)$-subharmonic function which is
$(\ell,q)$-suharmonic in $\ball_L(u)\subsetneq\ball$, with some $L\ge \ell\ge 0$
and $q\in(0,1)$. Then
\be\label{eq:lem.SubH.IID}
 f(u) \leq q^{ \left\lfloor \frac{L+1}{\ell+1} \right\rfloor} \CM(f, \ball)
 \leq q^{\frac{L-\ell}{\ell+1}} \CM(f, \ball),
 \quad\CM(f,  \ball) := \max_{x\in  \ball} |f(x)|.
\ee
\end{lemma}

%%%%%%%%%%%%%%%%%%%%%%%%%%%%%%%%%%%%%%%%%%%%%%%%%%%%%%%%%%%%%%%%%%%%%%%%%%%%%%%%%%%%
\begin{lemma}\label{lem:BiSubH.IID}
Let $ \ball\subset\DZ^d$ a finite connected subgraph,  $\ball_{r'}(u')$,
 $\ball_{r''}(u'') \subset\ball$ two disjoint balls with $|u'-u''|\ge r' + r'' +2$,
 and $f:\ball\times\ball\to \DR_+$, $(x',x'')\mapsto f(x',x'')$ a function
which is separately $(\ell,q)$-subharmonic in $x'\in\ball_{r'}(u')$
and in $x''\in\ball_{r''}(u'')$. Then
\be\label{eq:RDL.IID}
 f(u', u'') \leq q^{\left\lfloor \frac{r'+1}{\ell+1} \right\rfloor
 + \left\lfloor \frac{r''+1}{\ell+1} \right\rfloor } \CM(f, \ball)
 \leq q^{\frac{r' + r'' - 2\ell}{\ell+1} } \CM(f, \ball).
\ee
\end{lemma}

\begin{definition}\label{def:S.IID}
Given a sample $V(\cdot;\om)$, a ball $ \ball_L(u)$ is called $(E,m)$-non-singular ($(E,m)$-NS) if
for any pair of points $x,y\in \ball_L(u)$ with $|x-y|\ge L^{\frac{1+\rho}{\alpha}} = L^{7/8}$,
\be\label{eq:def.NS.IID}
 |\pt  \ball_L(u)|\,  \; |G(x,y;E;\om)|
 \leq e^{-\gamma(m,L) |x - y| + 2L^\beta},
\ee
where
\be\label{eq:def.gamma.IID}
 \gamma(m,L)  :=  m\left(1 + L^{-1/8}\right).
\ee
Otherwise, it is called $(E,m)$-singular ($(E,m)$-S).
\end{definition}

Observe that, with $m\ge 1$ and $L$ large enough,
\be\label{eq:gamma.7.8}
\bal
\gamma(m,L) |x - y| - 2L^\beta  
&= |x - y| m \left( 1 + L^{-\tau} - \frac{2 L^{\beta}}{m|x-y|} \right)
\\
& \ge |x - y| m \left( 1 + L^{-\tau} - 2 L^{\beta - 1 -\rho } \right)
\\
&\textstyle \ge  m \left( 1 + \half L^{-\tau} \right)|x - y|,
\eal
\ee
provided that $1+\rho - \beta > \tau$, which is the case with $\rho = 1/6$, $\beta=1/2$
and $\tau=1/8$. Therefore, the $(E,m)$-NS property \eqref{eq:def.NS.IID} implies that
\be\label{eq:def.NS.IID.implies}
 |\pt  \ball_L(u)|\,  \mymax{y\in\pt  \ball_L(u)} \; |G(u,y;E;\om)|
 \leq e^{-m \big(1 + \half L^{-\tau}\big) |x - y|}
\ee
for $x,y$ with $|x-y|\ge L^{\frac{1+\rho}{\alpha}} = L^{7/8}$.
\vskip2mm

In a more traditional MSA approach, the value of the exponent $m$ (often referred to as the "mass") depends upon the scale $L_k \sim (L_0)^{\alpha^k}$, $\alpha>1$, and the transition from $L_k$ to $L_{k+1}$  results in a tiny decay of the mass: $m_{k+1} = m_k(1 - o(1))$. Using the exponent $\gamma(m,L_k)$ allows to avoid rescaling the parameter $m$ itself and makes explicit the fact that the effective decay exponent is bounded from below by $m>0$, since $\gamma(m,L)> m$.

%%%%%%%%%%%%%%%%%%%%%%%%%%%%%%%%%%%%%%%%%%%%%%%%%%%%%%%%%%%%%%%%%%%%%%%%%%%%%%%%%%%%
\begin{lemma}\label{lem:cond.SubH.IID}
Consider a ball $\ball_L(u)$ and an operator $H=H_{\ball_L(u)}$ with fixed (non-random) potential $V$. Let $\{\psi_j, \, j=1, \ldots,  |\ball_L(u)|\}$ be the normalized eigenfunctions of $H$. Pick a pair of points $x',y' \in \ball_L(u)$ with $d(x',y') > 2(\ell + 1)$ and an integer
$R \in[\ell+2, \, d(x',y') - (\ell+2)]$. Suppose that any ball $\ball_{\ell}(v)$ with $v\in\ball_R(x')$ is $(E,m)$-NS, and set
\be\label{eq:lem:ConSubH.IID}
q =  q(\ell, m) = e^{-\gamma(m,\ell)\ell}
\ee
Then:
\begin{enumerate}[{\rm(A)}]
  \item the kernel $\Pi_{\psi_j}(x,y)$ of the spectral projection
  $\Pi_{\psi_j} = \ketbra{\psi_j}{\psi_j}$ is $(\ell,q)$-subharmonic in   $x \in \ball_R(x')$, with global maximum $\le 1$;
  \item if $\ball_L(u)$ is also $E$-NR, then the Green functions $G(x,y;E)=(H-E)^{-1}(x,y)$ are $(\ell,q)$-subharmonic in $x \in \ball_R(x')$, with global maximum $\le e^{L^\beta}$.
\end{enumerate}
\end{lemma}

%%%%%%%%%%%%%%%%%%%%%%%%%%%%%%%%%%%%%%%%%%%%%%%%%%%%%%%%%%%%%%%%%%%%%%%%%%%%%%%%%%%%%%%%%%%%%%%%%%%%
\subsection{Localization and tunneling}
\label{ssec:loc.tun}

\begin{definition}\label{def:loc}
Given a sample $V(\cdot;\om)$,  a ball $ \ball_L(u)$ is called
$m$-localized (\loc\,, in short) if any eigenfunction $\psi_j$ of the operator $H_{ \ball_L(u)}$ satisfies
\be\label{eq:def.mloc.IID}
|\psi_j(x)| \, |\psi_j(y)| \le e^{-\gamma(m,L) \|x - y\|}
\ee
provided that $\|x - y\|\ge L^{7/8}$.
Otherwise, it is called $m$-non-localized (\nloc).
\end{definition}

In terms of the parameter $\alpha$ (which is set to $4/3$ in Section \ref{sec:adapt.mixing}), the lower bound on the distance
$d(x,y) \ge L^{7/8}$ reads as  $d(x,y) \ge L^{\frac{1+\rho}{\alpha} }$, where
$\rho = (\alpha-1)/2 = 1/6$.

\begin{definition}\label{def:tun.IID}
Let real numbers $E$ and $m>0$ be given.
A ball $\ball_{\ell^{\alpha}}(u)$ is called
$m$-tunneling ($\mathtun$, in short), if it contains a pair of disjoint \nloc\, balls of radius $\ell$, and $m$-non-tunneling ($m$-NT), otherwise.
\end{definition}

Observe that, unlike the property of $E$-resonance or $(E,m)$-singularity, the tunneling property \textbf{is not related to a specific value of energy $E$}, and  even \textit{a single} tunneling ball has a small probability. This allows to adapt a very simple idea due to Spencer \cite{Spe88}  to a direct proof of localization in finite volumes.

%%%%%%%%%%%%%%%%%%%%%%%%%%%%%%%%%%%%%%%%%%%%%%%%%%%%%%%%%%%%%%%%%%%%%%%%%%%%%%%%%%%%%%%%%%%%%%%%%%%%
%%%%%%%%%%%%%%%%%%%%%%%%%%%%%%%%%%%%%%%%%%%%%%%%%%%%%%%%%%%%%%%%%%%%%%%%%%%%%%%%%%%%%%%%%%%%%%%%%%%%

\section{Simplified scale induction}
\label{sec:SimpleScaleInd}

\subsection{Initial scale bounds}
%%%%%%%%%%%%%%%%%%%%%%%%%%%%%%%%%%%%%%%%%%%%%%%%%%%%%%%%%%%%%%%%%%%%%%%%%%%
The following lemma (more precisely, its first assertion \eqref{eq:L0.GF.IID}) goes back to the works \cites{FMSS85,DK89}.

%%%%%%%%%%%%%%%%%%%%%%%%%%%%%%%%%%%%%%%%%%%%%%%%%%%%%%%%%%%%%%%%%%%%%%%%%%%%%%%%%%%%
\begin{lemma}[Initial scale bound: large disorder]\label{lem:L.0.IID}
For any $L_0>2$, $m>0$ and $p>0$ there exists $g_0 < \infty$ such that for all $g$ with
$|g|\ge g_0$, any ball $\ball_{L_{0}}(u)$ and any $E\in\DR$
\begin{eqnarray}
\pr{ \ball_{L_{0}}(u) \text{ is } (E,m)\text{\rm-S} } &\le L_{0}^{-p}
\label{eq:L0.GF.IID} \\
\pr{ \ball_{L_{0}}(u) \text{ is } \mathnloc } &\le L_{0}^{-p}.
\label{eq:L0.loc.IID}
\end{eqnarray}
\end{lemma}

\subsection{Tunneling and localization in finite balls}
\label{ssec:tun.and.loc.finite.balls}

%%%%%%%%%%%%%%%%%%%%%%%%%%%%%%%%%%%%%%%%%%%%%%%%%%%%%%%%%%%%%%%%%%%%%%%%%%%%%%%%%%%%
\begin{lemma}\label{lem:loc.and.NR.imply.NS.IID}
Let $E\in\DR$ and an \loc\, ball $ \ball_L(u)$ be given.
If $ \ball_L(u)$ is also  $E$-NR, then it is $(E,m)$-NS.
\end{lemma}
\proof
The matrix elements of the resolvent $G_\ball(E)$ read as follows:
\be\label{eq:G.basis}
G_\ball(x,y;E) = \sum_{E_j\in\sigma(H_{\ball_L(u)})}
 \frac{\overline{\psi_j(x)} \, \psi_j(y)}{E - E_j}.
\ee
If $\dist(E, \sigma(H_\ball)) \ge e^{-L^\beta}$ and
$\ln (2L+1)^d \le L^\beta$, then the \loc \, property implies
$$
\bal
\displaystyle |G_\ball(x,y;E)|  &
& \le e^{-\gamma(m,L)\|x - y\| + L^\beta + \ln |\ball|}
 \le e^{-\gamma(m,L)L + 2L^\beta}.
 \hfill\qed
\eal
$$

From this point on, we will work with a sequence of "scales" - positive integers $\{L_k, k\ge 0\}$, defined recursively by $L_{k+1} = [L_k^{\alpha}]+1$, $L_0>2$. In several arguments we will require the initial scale $L_0$ to be large enough.

%%%%%%%%%%%%%%%%%%%%%%%%%%%%%%%%%%%%%%%%%%%%%%%%%%%%%%%%%%%%%%%%%%%%%%%%%%%%%%%%%%%%
\begin{lemma}\label{lem:few.nloc.and.few.R.imply.nloc.IID}
There exists $\tilde L^{(1)}\in\DN$ such that for  $L_k\ge \tilde L^{(1)}$

\begin{enumerate}[\rm (A)]
  \item if a ball $ \ball_{L_k}(u)$ is $E$-CNR and contains no pair of disjoint $(E,m)$-S balls of radius $L_{k-1}$, then it is also $(E,m)$-NS;
   \item if a ball $ \ball_{L_k}(u)$ is $\mathntun$ and $E$-CNR, then it is also $(E,m)$-NS.
\end{enumerate}

\end{lemma}
\begin{proof}
\noindent
(A)
Fix two points $x,y\in\ball_L(u)$ with $|x-y|\ge L_{k-1}^{1+\rho}$.
By assumption,
\begin{enumerate}
  \item either $\ball_{L_k}(u)$ does not contain any $(E,m)$-S ball of radius $L_{k-1}$,
  \item or there is a ball $\ball_{L_{k-1}}(w)\subset\ball_{L_k}(u)$ such that any ball
  $\ball_{L_{k-1}}(v)\subset\ball_{L_k}(u)$ with $d(v, w)> 2L_{k-1}$ is $(E,m)$-NS.
\end{enumerate}
To treat both cases with one argument, in the simpler case (1) set, formally,
$w=y$ (although no exclusion is actually necessary).
Set
\be\label{eq:r.r.IID}
\bal
r'  &= \max(d(u, w) - L_{k-1}-1, 0), \\
r'' &= \max((y,  w) - L_{k-1}-1, 0).
\eal
\ee
Assume first that $r', r'' \ge L_{k-1}$. By triangle inequality,
\be\label{eq:r.r}
r' + r'' \ge |x-y| - 2L_{k-1} \ge L_{k-1}^{1+\rho} - 2L_{k-1}.
\ee

Consider the function  $f: \ball\times\ball \to \DR_+$ defined by
$
 f(x', x'') = |G_{\ball_{L_k}}(x', x''; E)|.
$
By construction, $f$ is  $(L_{k-1}, q_k)$-subharmonic both in
$x'\in\ball_{r'}(x)$ and in $x''\in\ball_{r''}(y)$, with
$
q \le e^{-\gamma(m, L_{k-1})},
$
since $\ball_{L_k}(u)$ is $E$-NR, and all  balls of radius $L_{k-1}$ both in $\ball_{r'}(x)$ and in $\ball_{r''}(y)$ are $(E,m)$-NS (being disjoint from $\ball_{L_{k-1}(w)}$). Collecting the assertion (B) of Lemma \ref{lem:cond.SubH.IID}, Lemma \ref{lem:BiSubH.IID} and inequality
\eqref{eq:r.r}, one can write, with the convention $-\ln 0 = +\infty$:
\be\label{eq:proof.lem:few.nloc.and.few.R.imply.nloc.IID}
\bal
-\ln f(x,y)   &
\textstyle
\ge
 -\ln \Bigg[
 \left( e^{-m(1 + \half L^{-\tau}_{k-1})L_{k-1}}\right)^{\frac{ R - 2L_{k-1} }{ L_{k-1} + 1} }
 e^{L_k^\beta}
 \Bigg]
\\
&
\textstyle
\ge m \left(1+ \half L_{k-1}^{-1/8} \right) \frac{L_{k-1}}{L_{k-1}+1} \left( R -2L_{k-1} \right)
- L_{k-1}^\beta
\\
&
\textstyle
\ge mR \left[ \left(1+ \half L_{k-1}^{-1/8} \right) \frac{L_{k-1}}{L_{k-1}+1} 
\left( 1 -\frac{2L_{k-1}}{R} \right)
- \frac{L_{k-1}^\beta}{ mR } \right]
\\
&
\textstyle
\ge R m \left( 1+ \half L_{k-1}^{-1/8} \right)  \left( 1 -  3L_{k-1}^{-1/6} \right)
\ge R m \left(1 + \quart L_{k-1}^{-1/8}   \right)
\\
 &> \gamma(m, L_k) \,|x - y| + \ln | \pt \ball_{L_k}(u)|,
\eal
\ee
as required for the $(E,m)$-NS property of the ball $\ball_{L_k}(u)$.

If $r'=0$ (resp., $r''=0$), the required bound follows from the subharmonicity of the function $f(x',x'')$ in $x''$ (resp., in $x'$).

\par\smallskip\noindent
(B) Assume otherwise. Then, owing to assertion (A), the $E$-CNR ball $\ball_{L_k}(u)$ must contain a pair of disjoint $(E,m)$-S  balls $\ball_{L_{k-1}}(x)$, $\ball_{L_{k-1}}(y)$. Neither of them is $E$-R, since  $\ball_{L_k}(u)$ is  $E$-CNR. By virtue of Lemma  \ref{lem:loc.and.NR.imply.NS.IID}, both $\ball_{L_{k-1}}(x)$ and $\ball_{L_{k-1}}(y)$ must be \nloc, so that $\ball_{L_k}(u)$ must be $\mathtun$, which contradicts the hypothesis.
\end{proof}

%%%%%%%%%%%%%%%%%%%%%%%%%%%%%%%%%%%%%%%%%%%%%%%%%%%%%%%%%%%%%%%%%%%%%%%%%%%%%%%%%%%%
\begin{lemma}\label{lem:NT.implies.nloc.IID}
There exists $\tilde L^{(2)}\in\DN$ such that, for all  $L_0\ge \tilde L^{(2)}$, if for any $E\in\DR$ a ball $\ball_{L_k}(u)$ contains no pair of disjoint $(E,m)$-S  balls of radius $L_{k-1}$, then it is $\mathloc$.
\end{lemma}
\proof

One can proceed as in the proof of the previous lemma, but with the functions
$
 f_j(x', x'') = |\psi_j(x') \psi_j(x'')|,
$
where $\psi_j, j\in[1, |\ball_{L_k}(u)|]$, are normalized eigenfunctions of operator $H_{\ball_{L_k}(u)}$. Notice that the $E_j$-non-resonance condition for $\ball_{L_k}(u)$ is not required here, since $\|\psi_j\|=1$, so the function $|f_j|$ is globally bounded by  $1$.
Let $x,y\in \ball_{L_k}(u)$ and assume that $d(x,y)=R \ge L_{k-1}^{1+\rho} = L_{k-1}^{7/6}$,
$\rho = 1/6$.

By assumption, either $\ball_{L_k}(u)$  contains no $(E_j,m)$-S $L_{k-1}$-ball, or there is a ball $\ball_{L_{k-1}}(w)$ such that any ball $\ball_{L_{k-1}}(v)$ disjoint with $\ball_{L_{k-1}(w)}$ is $(E_j,m)$-NS. We consider first the latter case and prove the assertion, avoiding the ball $\ball_{L_{k-1}}(w)$. To this end, set
$r' = d(x, w)$, $r'' = d(y, w)$. Then the function $(x',x'')\mapsto \psi_j(x') \psi_j(x'')$ is $(\ell,q)$-subharmonic in $\ball_{r'}(x)\times\ball_{r''}(y)$. Therefore, one can apply the bound \eqref{eq:RDL.IID}, with
$r'+r'' \ge R - 2L_{k-1}$,
and write, with the convention $-\ln 0 = +\infty$:
\be\label{eq:proof.lem.NT.implies.nloc.IID}
\bal
-\ln |\psi_i(x) \psi(y) |  &\ge
 -\ln \left( e^{-\gamma(m,L_{k-1})L_{k-1}}\right)^{\frac{ R - 2L_{k-1} }{ L_{k-1} + 1} }
%\\
\eal
\ee
A direct comparison of Eqn \eqref{eq:proof.lem:few.nloc.and.few.R.imply.nloc.IID}
with Eqn  \eqref{eq:proof.lem.NT.implies.nloc.IID}
shows that the RHS of the latter is bigger, owing to the absence of the
factor $e^{L_{k-1}^\beta} >1$, thus it admits the same lower bound
as in Eqn \eqref{eq:proof.lem:few.nloc.and.few.R.imply.nloc.IID}.
\qedhere

\begin{remark}
\label{rem:universality.IID}
Assertions of Lemma \ref{lem:few.nloc.and.few.R.imply.nloc.IID} and Lemma \ref{lem:NT.implies.nloc.IID} are  deterministic and do not rely upon a particular structure of the potential. In other words, these statements are valid for  arbitrary LSO, including multi-particle operators.
\end{remark}

%%%%%%%%%%%%%%%%%%%%%%%%%%%%%%%%%%%%%%%%%%%%%%%%%%%%%%%%%%%%%%%%%%%%%%%%%%%%%%%%%%%%

\begin{lemma}[\textbf{Main Inductive Lemma}]\label{lem:induct.IID}
Let $p, b>0$ satisfy
\be\label{eq:cond.p.b}
\textstyle
p > \frac{2\alpha^2 d}{2-\alpha^2}, \;
0 < 3b \le \frac{2 - \alpha^2}{\alpha^2} - \frac{2d}{ p },
\ee
Suppose that for some $k\ge 0$ and all $0 \le k' \le k$ the following bound holds true:
\be\label{eq:lem.induct.IID}
\forall\, u\in\DZ^d\qquad \pr{ \text{ $\ball_{L_{k'}}(u)$ is } \mathnloc }
\le L_{k'}^{-p(1+b)^{k'}}.
\ee
Then, for  $L_0>0$ large enough,
$\forall\, u\in\DZ^d$
\be\label{eq:lem.induct.next.Tun.IID}
\ba
\pr{ \text{  $\ball_{L_{k+1}}(w)$ is } \mathnloc } \le \quart L_{k+1}^{-p(1+b)^{k+1}}
\ea
\ee
and for any pair of disjoint balls $\ball_{L_{k}}(u)$, $\ball_{L_{k}}(v)\subset\ball_{L_{k+1}}(w)$
\be\label{eq:lem.induct.next.S.IID}
\pr{\exists\,E:\; \ball_{L_{k}}(u) \text{ and } \ball_{L_{k}}(v)
\text{ are $(E,m)$-S  } }
% \le  \half L_k^{-\frac{2p(1+b)^k}{\alpha^2}+2d}
\le \half L_k^{-p(1+b)^{k+1}}.
\ee
\end{lemma}
\proof
Consider the following events:
$$
\bal
\CN_{k+1} &:= \{ \text{$\ball_{L_{k+1}}(w)$ is } \mathnloc \} \\
\CS_{k}^{(2)} &:= \{\exists\,E\;\exists\,  \text{ disjoint $(E,m)$-S balls } \ball_{L_{k}}(u),  \ball_{L_{k}}(v)\subset \ball_{L_{k+1}}(w) \} \\
\CR_{k}^{(2)} &:= \{\exists\,E\;\exists\,  \text{ disjoint $E$-PR balls } \ball_{L_{k}}(u),  \ball_{L_{k}}(v)\subset \ball_{L_{k+1}}(w)  \}.
\eal
$$
By Lemma \ref{lem:NT.implies.nloc.IID},
$$
\CN_{k+1} \subset \CR_{k}^{(2)} \cup \CS_{k}^{(2)} \equiv
  \CR_{k}^{(2)} \cup \left(\CS_{k}^{(2)}\setminus \CR_{k}^{(2)} \right),
$$
and by Wegner-type bound \eqref{eq:W2.IID}, $\pr{\CR_k^{(2}}\le e^{-L_k^{\beta'}}$, so that it remains to bound $\pr{\CS_{k}^{(2)}\setminus \CR_{k}^{(2)} }$. Fix points $u,v$ with
$d(u, v)>2L_k$ and introduce the event
\be\label{eq:S2.u.v}
\CS_{k}^{(2)}(u,v) = \{\exists\,E:\;
\text{ $\ball_{L_{k}}(u)$ and $\ball_{L_{k}}(v)$  are $(E,m)$-S} \}.
\ee
Within the event $\CS_{k}^{(2)}(u,v)\setminus \CR_{k}^{(2)}$, either
$\ball_{L_k}(u)$ or $\ball_{L_k}(v)$ must be $E$-CNR; without loss of generality, suppose
$\ball_{L_k}(u)$ is $E$-CNR. At the same time, it is $(E,m)$-S, so by assertion (B) of Lemma \ref{lem:few.nloc.and.few.R.imply.nloc.IID}, $\ball_{L_k}(u)$ must  be $\mathtun$, i.e., it must contain a pair of disjoint \nloc\, balls $\ball_{L_{k-1}}(y')$, $\ball_{L_{k-1}}(y'')$. Using the inductive assumption \eqref{eq:lem.induct.IID} and  independence\footnote{Clearly, a weak dependence would suffice for this argument; cf. Section \ref{sec:adapt.mixing}.}
 of  random operators
$H_{\ball_{L_{k-1}}(y')}$, $H_{\ball_{L_{k-1}}(y'')}$, we can write
$$
\pr{ \text{$\ball_{L_{k-1}}(y')$ and $\ball_{L_{k-1}}(y')$ are \nloc} }
\le L_{k-1}^{-2p(1+b)^{k-1}} \le  L_{k+1}^{- \frac{2p}{\alpha^2} (1+b)^{k-1} }.
$$
The number of all pairs $y', y''$ is bounded by $|\ball_{L_{k+1}}(u)|^2/2$, so that
\be\label{eq:prob.S2.not.R2}
\textstyle
\pr{ \CS_{k}^{(2)} \setminus \CR_{k}^{(2)} }
\le \half (2L_{k+1} + 1)^{2d} L_{k+1}^{- \frac{2p}{\alpha^2} (1+b)^{k-1}}
< \quart  L_{k+1}^{-p(1+b)^{k+1}},
\ee
for $L_0$ (hence, $L_{k+1}$) large  enough,  provided that
\be\label{eq:cond.b}
\textstyle
\frac{2}{\alpha^2} - \frac{2d}{ p(1+b)^{k-1} } > (1+b)^{2}.
\ee
Observe that for $b\in(0,1)$, $(1+b)^2 < 1+3b$ and
$\frac{2d}{ p } > \frac{2d}{ p(1+b)^{k-1} }$. Therefore, if $p$ and $b$ fulfill
the conditions \eqref{eq:cond.p.b}, 
%$$
%p > \frac{2\alpha^2 d}{2-\alpha^2}, \;
%0 < 3b \le \frac{2 - \alpha^2}{\alpha^2} - \frac{2d}{ p },
%$$
then \eqref{eq:cond.b} is also fulfilled, yielding the last inequality in \eqref{eq:prob.S2.not.R2}. With $\alpha = 4/3$, this reads as $p>16d$ and
$b< \frac{1}{24} - \frac {2d}{3p}$. Therefore, with  $L_0$ large enough,
\begin{equation*}
\qquad\qquad\qquad\qquad\qquad\qquad
\textstyle
\pr{ \CN_{k+1} }  <  \quart L_{k+1}^{-p(1+b)^{k+1}}.
\qquad\qquad\qquad\qquad\qquad\qquad{}_{{}_{\qed}}
\end{equation*}
%\qedhere
\vskip2mm

%%%%%%%%%%%%%%%%%%%%%%%%%%%%%%%%%%%%%%%%%%%%%%%%%%%%%%%%%%%%%%%%%%%%%%%%%%%
\begin{theorem}\label{thm:ind.loc.1p.IID}
For any $m>0$, $p>16d$ and $L_0>2$ there exist $g_0 < \infty$ and $b >0$ such that for any $g$ with $|g|\ge g_0$,  all $k\ge 0$ and any ball $\ball_{L_{k}}(u)\subset \DZ^d$,
\be\label{eq:thm.ind.1p.loc.IID}
\pr{ \text{ $\ball_{L_{k}}(u)$ is } \mathnloc } \le L_{k}^{-p(1+ b )^k}.
\ee
\end{theorem}

\proof
The claim follows by induction from Lemma \ref{lem:induct.IID} and Lemma \ref{lem:L.0.IID}.
\qedhere

\medskip
{\sl
Theorem \ref{thm:ind.loc.1p.IID}  establishes the \textbf{exponential localization of all eigenfunctions} of operator $H_{\ball_\CL}(\om)$ in an arbitrarily large ball $\,\ball_\CL$ with high probability. The lower bounds on the eigenfunction decay exponent $m>0$ as well as the decay exponent $p>0$ for  the probability in \eqref{eq:thm.ind.1p.loc.IID}, are uniform in $\CL\ge L_0$. This makes the statement of Theorem \ref{thm:ind.loc.1p.IID} sufficient for applications to physical models of disordered media of arbitrarily large size, whether it is a 45-{\rm nm} film of diameter $\sim$ 2{\,\rm mm} (approx. $10^{10}$ lattice bond units, corresponding  to the size of a modern CPU chip, which requires 2-3 steps of scale induction) or a sample of the size of the Milky Way (depending on the initial scale $L_0$, it may require from 5-6 to 10-12 scaling steps).
}
\medskip

In the next subsection we translate the results of previous sections into the language of eigenfunction correlators. Unlike the Fractional-Moment Method, the Multi-Scale Analysis does not provide exponential decay of EF correlators; usually one proves a polynomial decay with a fixed exponent. Owing to a stronger probabilistic bound of "unwanted" events in finite balls, of the form $L_k^{-p(1+b)^k}$, we will be able to prove a slightly stronger decay bound for the EF correlators (cf. Theorem \ref{thm:DL.IID}).

\subsection{Strong dynamical localization in finite volumes}
\label{ssec:SimpleDL.finite.L}

Now we will derive uniform upper bounds on EF correlators in finite, but arbitrarily large balls
from the MSA bounds, using a simplified version of the Germinet--Klein argument
\cite{GK01}. Recall that originally the implication "MSA $\Rightarrow$ DL" has been
proven by Germinet--De Bi\`{e}vre \cite{GD98} and by Damanik--Stollmann \cite{DS01}
(strong dynamical localization). Formally, Germinet and Klein \cite{GK01} considered differential
operators in $\DR^d$, but an adaptation of their technique to lattice models is not
difficult. Moreover, it becomes quite elementary when operators in
finite balls are considered. Generally speaking, it suffices that finite-volume
operators have compact resolvent; on the lattice, the operators $H_{\ball_L(u)}$ are even
finite-dimensional and have a finite orthogonal eigenbasis. This allows to avoid an analysis of Hilbert--Schmidt norms of their spectral projections (inevitable in the entire
lattice/Euclidean space) and to replace it with an elementary application of Bessel's
inequality.

Denote by $\csB_1(I)$ the set of all Borel functions $\phi:\DR\to\DC$ with
${\rm supp}\,\phi\subset I$ and $\|\phi\|_\infty \le 1$.

\begin{theorem}\label{thm:DL.IID}
 Fix an integer $L\in\DN^*$ and assume that  the following bound holds for any pair of disjoint balls $\ball_L(x), \ball_L(y)$:
$$
\pr{ \exists\, E\in I:\, \text{ $\ball_L(x)$ and $\ball_L(y)$ are $(E,m)$-S} } \le f(L).
$$
Then for any $x,y\in\DZ^d$ with $\rd(x,y)> 2L+1$, any connected subset
$\Lam\supset\ball_L(x) \cup \ball_L(y)$ and any Borel function $\phi\in\csB_1(I)$
\be\label{eq:thm.MSA.to.DL}
\esm{ \langle \big|\one_{x} | \phi(H_\Lam(\om)) | \one_{y} \rangle \big| }
\le CL^d \eul^{-mL} + f(L).
\ee
\end{theorem}

\proof Fix points $x,y\in\DZ^d$ with $\rd(x,y)> 2L+1$ and a finite connected graph
$\Lam\supset\ball_L(x) \cup \ball_L(y)$. The operator $H_\Lam(\om)$
has a finite orthonormal eigenbasis $\{\psi_i\}$ with respective eigenvalues
$\{\lam_i\}$. Set $S = \pt \ball_L(x) \cup \pt \ball_L(y)$ (recall: this is a set of \emph{pairs}
$(u,u')$). Suppose that for some $\om$, for each $i$ there is $z\in \{x,y\}$ such that
$\ball_L(z_i)$ is $\lam_i,m)$-NS; let $\{v_i \}= \{x,y\}\setminus \{z_i\}$.
Denote
$\mu_{x,y}(\phi) = \langle \big|\one_{x} | \phi(H_\Lam(\om)) | \one_{y} \rangle$,
with $\big|\mu_{x,y}(\phi)\big|\le 1$.
Then by the GRI for eigenfunctions,
and by Bessel's inequality used at the last stage of derivation,
$$
\bal
 \big|\mu_{x,y}(\phi)\big|
& \le \|\phi\|_\infty \, \sum_{\lam_i \in I} |\psi_i(x) \psi_i(y)|
\le \sum_{\lam_i \in I} |\psi_i(z_i) \psi_i(v_i)|
\\
& \le \sum_{\lam_i \in I} |\psi_i(v_i)| \eul^{-mL}
      \sum_{(u,u')\in\pt \ball_L(z_i)} |\psi_i(u)|
\qquad\qquad\qquad\qquad\qquad\qquad
\\
%%
%\eal
%$$
%$$
%\bal
& \le \eul^{-mL} \sum_{\lam_i \in I} \;\sum_{(u,u')\in S}
      |\psi_i(u)| \left(|\psi_i(x)| + |\psi_i(y)|  \right)
%\\
%
\eal
$$
$$
\bal
&\le \eul^{-mL} |S| \, \max_{u\in\Lam} \sum_{\lam_i \in I}
    \half \left( |\psi_i(u)|^2 + |\psi_i(x)|^2 + |\psi_i(y)|^2 \right)
\\
& \le  \eul^{-mL} \frac{|S|}{2} \, \max_{u\in\Lam}
       \left( 2\|\one_u\|^2 + \|\one_x\|^2 + \|\one_y\|^2 \right)
%\\
%
%&
= \eul^{-mL} |S| \cdot 2
\eal
$$
where $|S|\le C L^d$.
Denote
$\CN_L = \myset{\exists\, E\in I:\, \text{ $\ball_L(x)$ and $\ball_L(y)$ are $(E,m)$-S}}$,
with $\pr{\CN_L}\le f(L)$, by assumption. Further,
$$
\esm{ \mu_{x,y}(\phi) } = \esm{ \one_{\CN_L}\mu_{x,y}(\phi) }
   + \esm{ \one_{\CN_L}\mu_{x,y}(\phi) }  \le f(L) + 2CL^d \eul^{-mL}.
\qedhere
$$

\qedhere
\vskip2mm

\subsection{Strong dynamical localization on the entire lattice}
\label{ssec:SimpleDL.lattice}

Here we follow the same strategy as in earlier works by Aizenman et al.
\cite{A94,ASFH01}.

\begin{theorem}\label{thm:DL.lattice.IID}
Consider the Hamiltonian $H(\om)$ of the form \eqref{eq:H} with random potential satisfying the assumption \Wone, and fix an interval $I\subset\DR$.  There exists a $L_0^* < \infty$ with the following property:
if the conditions \eqref{eq:L0.GF.IID} and \eqref{eq:L0.loc.IID} hold for some  $L_0 \ge L_0^*$, then for all $x,y\in\DZ^d$, $x\ne y$,  and some $c, a>0$,
\be\label{eq:DL.bound.any.s.lattice.IID}
\esm{ \sup_{\|f\|_{\infty}\le 1} \,
\Big| \langle x \,|\, f(H_{}(\omega)) \BP_I(H_{}(\omega)) \,|\, y\rangle \Big| }
 \le \Const\, \|x-y\|^{-a \ln^c \|x-y\| }.
\ee
\end{theorem}

\proof
It suffices to use an argument given earlier in \cite{ASFH01}.
For any  ball $\ball$ and any points $x,y\in\ball$ introduce a spectral measure $\mu^{x,y}_{\ball,\om}$ uniquely defined, for any bounded Borel function $f$, by
$$
\int\, f(\lambda)\, d\mubxy(\lambda)
= \bra{\delta_x} f(H_\ball(\omega)) \BP_I(H_{\ball}(\omega)) \ket{\delta_y},
$$
and similar spectral measures $\muxy (\equiv \mu^{x,y}_{\DZ^d,\om})$ for the operator $H(\om)$ on the entire lattice.
Then $\mukxy$ converge vaguely to $\muxy$ as $k\to\infty$, so that by virtue of Fatou lemma on convergent measures, for any measurable set $\CE\subset\DR$
$$
\esm{ |\muxy| (\CE) } \le \liminf_{k\to\infty} \; \esm{ |\mukxy| (\CE) }.
$$
Taking functions $f_t: \lambda\mapsto e^{it\lambda}$, $t\in\DR$, we see that the uniform bounds on dynamical localization in finite volumes $\ball_{L_k}(0)$, established in the previous sections, imply the dynamical localization on the entire lattice.
\qedhere

Theorem \ref{thm:DL.lattice.IID} leads directly to the following, more traditional form of dynamical localization. Let $\B{X}$ be a multiplication operator defined by
$(\B{X}f)(x):= (1+\|x\|)f(x)$.

\begin{theorem}
Under the assumptions of Theorem \ref{thm:DL.lattice.IID}, there exist $a,c>0$ such that
for any finite subset $K\subset\DZ^d$ and any finite interval $I\subset\DR$
\begin{equation}\label{eq:thm.DL.lattice.IID.propagator}
\expect\left[ \sup_{t\in\D{R}} \;\left\| e^{a \ln^{1+c}\B{X}} \, \eul^{-\ii tH(\omega)}
          P_{I }(H(\omega))
\one_{K}\right\|\right] < \infty.
\end{equation}
\end{theorem}

\subsection{Exponential decay of eigenfunctions on the entire lattice}
\label{ssec:SimpleSL.lattice}

The dynamical localization is known to imply pure point spectrum, owing to RAGE theorem(s); see the original papers \cite{R69}, \cite{AG73}, \cite{E78} and their detailed discussion in \cite{CFKS87}. This allows to consider in  Theorem \ref{thm:Main.SL} below, from the beginning, a square summable (hence, bounded) eigenfunction $\psi$ on the lattice $\DZ^d$, avoiding a usual reference to a Shnol-type theorem stating that spectrally a.e. generalized eigenfunction is polynomially bounded.
The general strategy goes back to \cite{DK89}; using the "Radial Descent lemma" (Lemma \ref{lem:SubH.IID}) and making a small concession in probability bounds (which experts in MSA may notice) results in a shorter and more transparent proof.

\begin{theorem}\label{thm:Main.SL}
For $\DP$-a.e. $\om\in\Om$ every normalized eigenfunction $\psi$
of operator $H(\om)$ satisfies the following bound: for some $R(\om)$ and all
$y$ with $\|y\|\ge R(\om)$
\be\label{eq:thm.Main.SL}
  |\psi(y)| \le e^{-m\|y\|}.
\ee
\end{theorem}
\proof
By  Borel--Cantelli lemma combined with \eqref{eq:lem.induct.next.S.IID}, there is a subset $\Om'\subset\Om$ with $\pr{\Om'}=1$ such that for any $\om\in\Om'$ and some $k_0(\om)$,  all
$k\ge k_0$ and any $E\in\DR$ there is no pair of disjoint $(E,m)$-S balls $\ball_{L_k}(x), \ball_{L_k}(y)\subset \ball_{2L_{k+2}}(0)$. Fix $\om\in\Om'$.

Let $\psi_n$ be an eigenfunction of $H(\om)$ with eigenvalue $E_n$.
If $\hx_n\in\ball_{L_{k-1}}(0)$, then $\ball_{L_k}(0)$ is $(E_n,m)$-S, so any ball $\ball_{L_k}(y)\subset \ball_{2L_{k+2}}(0)$ with $\|y\|>2L_k$ is $(E_n,m)$-NS.

For any $y$ with $\|y\| \ge L_{k_0+1}$ there is $k\ge k_0$ such that $\|y\|\in[L_{k+1}, L_{k+2})$. Fix $y$, set $R = \|y \| - 2L_k-1$ and observe that the function $x\mapsto |\psi_n(x)|$ is $(L_k,q)$-subharmonic in $\ball_{R}(y)$,
with $q = e^{-\gamma(m,L_k)L_k}$. Now Lemma \ref{lem:SubH.IID}
implies, for $L_k$ large enough,
$$
 \frac{\ln |\psi(y)|}{\|y\|}
\le -m \left(1+L_k^{-1/8} \right) \left(1 - \frac{2L_k+1}{\|y\|}\right)
\le -m \left(1 + \half L_k^{-1/8} \right) < - m
$$
leading to the assertion \eqref{eq:thm.Main.SL}.
\qedhere

%%%%%%%%%%%%%%%%%%%%%%%%%%%%%%%%%%%%%%%%%%%%%%%%%%%%%%%%%%%%%%%%%%%%%%%%%%%%%%%%%%%%%%%%%%%%%%%%%
%%%%%%%%%%%%%%%%%%%%%%%%%%%%%%%%%%%%%%%%%%%%%%%%%%%%%%%%%%%%%%%%%%%%%%%%%%%%%%%%%%%%%%%%%%%%%%%%%
%%%%%%%%%%%%%%%%%%%%%%%%%%%%%%%%%%%%%%%%%%%%%%%%%%%%%%%%%%%%%%%%%%%%%%%%%%%%%%%%%%%%%%%%%%%%%%%%%

\section{Adaptation to low-energy analysis at weak disorder}
\label{sec:adapt.low.E}

If the amplitude of the random potential $V(x;\om)$ is  small, the existing methods allow to establish Anderson localization  only for "extreme" energies. For example, if the (sharp) lower edge of the random potential is given by $E^0 > -\infty$, localization can be established in a narrow interval $I = [E^0, E^0 + \eta]$, with sufficiently small $\eta>0$. Then representation \eqref{eq:G.basis} can no longer be used; it is more convenient to start with the analysis of resolvents
and  modify the notion of "tunneling" balls as follows.

\begin{definition}\label{def:tun.low.E}
Let  an interval $I\subset \DR$ and a number $m>0$ be given.
A ball $ \ball_{\ell^{\alpha}}(u)$ is called $(m,I)$-tunneling ($(m,I)$-T) if, for some $E\in I$, it contains a pair of disjoint $(E,m)$-S balls of radius $\ell$.
Otherwise, it is called $(m,I)$-non-tunneling ($(m,I)$-NT).
\end{definition}

%%%%%%%%%%%%%%%%%%%%%%%%%%%%%%%%%%%%%%%%%%%%%%%%%%%%%%%%%%%%%%%%%%%%%%%%%%%%%%%%%%%%
\begin{lemma}[Combes--Thomas estimate]\label{lem:CT}
Consider a lattice Schr\"{o}\-dinger operator $H_\Lambda$ on a subset
$\Lambda\subset \mathbf{Z}^d$. Suppose that for some $E\in\DR$,
$\dist(E, \Sigma(H_\Lambda))\ge \eta>0$. Then for all $x,y\in\Lambda$
\begin{equation}\label{eq:CT.BCH}
|(H - E)^{-1}(x,y)| \le 2\eta^{-1} e^{ - \frac{\eta}{5d} \|x-y\|}.
\end{equation}
\end{lemma}

The proof  of Combes--Thomas estimate \cite{CT} for lattice models can be found, e.g., in
\cite{K07} where it is shown that one can take the exponent  $\frac{\eta}{12d}$. A minor
improvement of the argument from \cite{K07} allows to obtain $\frac{\eta}{5d}$.
(In Eqn (11.10) from \cite{K07}, one can use the inequality $\frac{2}{5}e^{1/5}< \half$
instead of a more generous bound $e^{1/12}< e^1 < 3$.)

%%%%%%%%%%%%%%%%%%%%%%%%%%%%%%%%%%%%%%%%%%%%%%%%%%%%%%%%%%%%%%%%%%%%%%%%%%%%%%%%%%%%
\begin{lemma}\label{lem:CT.prob}
Consider random LSO
$H^{(\BN)}_{\ball_\ell(u)}(\omega) = \Delta + V(x;\omega)$ with Neumann boundary conditions. Suppose that the random variables $V(x;\omega)$ are IID, non-negative and non-constant, and that for some $\eta>0$
and $c>0$, all $u\in\DZ^d$  and all  $\ell \ge \ell_0 >0$
\begin{equation}\label{eq:cond.V.CT}
\pr{ |\ball_\ell|^{-1}\, \sum_{x\in \ball_\ell} V(x;\omega)\leq 2\eta }
\leq e^{-c|\ball_\ell|}.
\end{equation}
Then for some $C>0$, $L_0\in\DN$ the lowest eigenvalue
$E^{(\BN)}_0(\omega)$  of $H^{(\BN)}_{\ball_{L_0}(u)}(\omega)$ satisfies
\begin{equation}
\pr{ E^{(\BN)}_0(\omega) \le  2C L_0^{-1/2} } \leq e^{-c|\ball_\ell(u)|^{1/4}}.
\end{equation}
By Dirichlet--Neumann bracketing, the same bound holds true for the lowest eigenvalue of the LSO
$H^{(\BD)}_{\ball_{L_0}(u)}(\omega)$ with Dirichlet boundary conditions.
\end{lemma}

A detailed discussion of the Lifshitz tails phenomenon, along with all ingredients of the proof of Lemma \ref{lem:CT.prob}, can be found, e.g., in Ref \cite{K07}.

Using Lemma \ref{lem:CT} and  Lemma \ref{lem:CT.prob}, we come to the following
\begin{corollary}\label{cor:CT.prob}
Under the assumptions of Lemma \ref{lem:CT.prob}, there exist $\tilde  L^{(3)}\in\DN$ and
$C, c>0$ such that, for any $L_0\ge \tilde L^{(3)}$ and for some $\eta(L_0)>0$,
$m(L_0) \ge CL_0^{-1/2}$,
\begin{equation}\label{eq:cor.CT.S}
\pr{ \exists\, E\in[E^0, E^0+ \eta(L_0)] :\, \ball_{L_0}(0) \text{ is } (E,m(L_0))-S }
\leq e^{-c L_0^{d/4}},
\end{equation}
where $E^0$ is the lower edge of the spectrum of LSO $\Delta + V(x;\om)$ on the lattice $\DZ^d$.
\end{corollary}

%%%%%%%%%%%%%%%%%%%%%%%%%%%%%%%%%%%%%%%%%%%%%%%%%%%%%%%%%%%%%%%%%%%%%%%%%%%%%%%%%%%%

The next statement is merely a reformulation of Lemma \ref{lem:NT.implies.nloc.IID} for energies $E$ restricted to an interval $I\subset \DR$.
\begin{lemma}\label{lem:CT.cor.loc}
Let an interval $I\subset \DR$ be given, and suppose that a ball $\ball_{L_0}(u)$ is $(m,I)$-NT. Then it is also \Iloc.
\end{lemma}

\begin{corollary}\label{cor:CT.cor.loc}
Under the assumptions of Lemma \ref{lem:CT.prob}, there exist $\tilde L^{(4)}\in\DN$ and
$C, c>0$ such that, for any $L_1\ge \tilde L^{(4)}$ and for some $\eta(L_1)>0$,
$m=m(L_1) \ge CL_1^{-1/3}$,
\begin{equation}
\pr{ \ball_{L_1}(0) \text{ is \Inloc\, } } \leq e^{-c L_1^{d/6}}.
\end{equation}
\end{corollary}

\proof
First, we choose $L_0$ as in Corollary \ref{cor:CT.prob}, and set 
$L_1 = \big[L_0^{\alpha}\big]+1 = \big[L_0^{4/3} \big]+1$, so that $L_0 \approx L_1^{3/4}$. Owing to \eqref{eq:cor.CT.S}, with probability not smaller than
$$
 1 - |\ball_{L_1}(0)|^2 e^{-2cL_0^{d/4}} \ge 1 - e^{-cL_1^{3d/16}}
$$
there is no pair of disjoint $(E,m)$-singular balls of radius $L_0$ inside $\ball_{L_1}(0)$. Now the claim follows from Lemma \ref{lem:CT.cor.loc}.
\qedhere

Corollary \ref{cor:CT.cor.loc} allows to establish uniform bounds on eigenfunction correlators. However, the main  technical tool of the scale induction at "extreme" energies becomes the following analog of  Lemma \ref{lem:NT.implies.nloc.IID} for the Green functions:

%%%%%%%%%%%%%%%%%%%%%%%%%%%%%%%%%%%%%%%%%%%%%%%%%%%%%%%%%%%%%%%%%%%%%%%%%%%%%%%%%%%%
\begin{lemma}\label{lem:CT.ind.step}
For any $C>0$ there exists $L^*(C)>0$ such that for any $L_k\ge L^*(C)$, if a ball $\ball_{L_k}$ is $E$-NR and $(m,I)$-NT with $m\ge CL_{k}^{-1/2}$, then it is $(E,m)$-NS.
\end{lemma}

The proof is  quite similar to the proof of Lemma \ref{lem:NT.implies.nloc.IID}, with minor modifications. Actually, the sufficiency of the lower bound  $m\ge CL_{k}^{-1/2}$ on the decay exponent $m$ for the purposes of the MSA is a well-known fact.

%%%%%%%%%%%%%%%%%%%%%%%%%%%%%%%%%%%%%%%%%%%%%%%%%%%%%%%%%%%%%%%%%%%%%%%%%%%%%%%%%%%%
\begin{lemma}[\textbf{Main inductive lemma} for an energy band $I$]\label{lem:induct.IID.low.E}
Let an interval $I\subset \DR$ be given. Suppose that for some $k\ge 0$, $C>0$, $m\ge CL_k^{-1/2}$ and  $p > 16d$, the following bound holds true:
% for any ball $\ball_{L_k}(w)$:
\be\label{eq:lem.induct.IID.low.E}
\forall\, w\in\DZ^d\qquad \pr{ \text{ $\ball_{L_{k}}(w)$ is  $(m,I)$-T } } \le L_{k}^{-p},
\ee
Then, for $L_0>0$ large enough, there exists $b>0$ such that for all $u\in\DZ^d$
\be\label{eq:lem.induct.next.Tun.IID.low.E}
\pr{ \text{  $\ball_{L_{k+1}}(u)$ is $(m,I)$-T } } \le L_{k+1}^{-p(1+b)}.
\ee
\end{lemma}

The proof repeats almost verbatim that of  Lemma \ref{lem:induct.IID}, so we omit it here.

Now one can conclude as in the case of large disorder and prove dynamical localization in the energy band $I\subset\DR$ on the entire lattice; cf. Section \ref{ssec:SimpleDL.lattice}.

%%%%%%%%%%%%%%%%%%%%%%%%%%%%%%%%%%%%%%%%%%%%%%%%%%%%%%%%%%%%%%%%%%%%%%%%%%%%%%%%%%%%%%%%%%%%%%%%%
%%%%%%%%%%%%%%%%%%%%%%%%%%%%%%%%%%%%%%%%%%%%%%%%%%%%%%%%%%%%%%%%%%%%%%%%%%%%%%%%%%%%%%%%%%%%%%%%%
%%%%%%%%%%%%%%%%%%%%%%%%%%%%%%%%%%%%%%%%%%%%%%%%%%%%%%%%%%%%%%%%%%%%%%%%%%%%%%%%%%%%%%%%%%%%%%%%%

% mixing CASE

%%%%%%%%%%%%%%%%%%%%%%%%%%%%%%%%%%%%%%%%%%%%%%%%%%%%%%%%%%%%%%%%%%%%%%%%%%%%%%%%%%%
%%%%%%%%%%%%%%%%%%%%%%%%%%%%%%%%%%%%%%%%%%%%%%%%%%%%%%%%%%%%%%%%%%%%%%%%%%%%%%%%%%%
\section{Adaptation to correlated random potentials }
\label{sec:adapt.mixing}

In this section we assume that the random field $V$ fulfills the conditions \Wtwo--\Wthree.
Note that we consider here only the case of large disorder, in order to use a more streamlined approach from Section \ref{sec:SimpleScaleInd}, but an adaptation to the low-energy analysis, close to that described in Section \ref{sec:adapt.low.E}, is fairly straightforward.

\subsection{Resonant and  singular balls}

The following statement is an adaptation of the Wegner-type bound to correlated potentials satisfying the conditions \Wtwo--\Wthree.

%%%%%%%%%%%%%%%%%%%%%%%%%%%%%%%%%%%%%%%%%%%%%%%%%%%%%%%%%%%%%%%%%%%%%%%%%%%%%%%%%%%%
\begin{lemma}[Wegner-type bound for correlated potentials]\label{lem:WS.corr} Fix a finite interval $I\subset\DR$.
Under the assumption \Wtwo, there exists $L^*>0$ and $\beta' \in(0, \beta)$ such that for all $L\geq L^*$ and any  ball $\ball_{L}(x)$
\begin{eqnarray}
\forall\, E\in \DR\quad
\pr{  \text{  $\ball_{L}(x)$ is $E$-\rm{R} } } \le e^{-L^{\beta}}
\label{eq:W1.mixing} \\
\forall\, E\in \DR\quad
\pr{  \text{  $\ball_{L}(x)$ is not $E$-\rm{CNR} } } \le e^{-L^{\beta'}}.
\label{eq:W1.CNR.mixing}
\end{eqnarray}
As a result, under the assumption \Wthree,  for all $L$ large enough and any pair of $L$-distant balls $\ball_{L}(x)$, $\ball_{L}(y)$
\be\label{eq:W2.mixing}
\pr{ \exists\, E\in I:\, \text{ neither $\ball_{L}(x)$ nor $\ball_{L}(y)$  is $E$-\rm{CNR} } }
\le e^{-c\ln^2 L}.
\ee
\end{lemma}
\proof
The bound \eqref{eq:W1.mixing} follows directly from an extension of Stollmann's lemma (cf. \cite{St01}) on monotone functions to correlated random fields, given in our earlier work \cite{C08}. The bounds \eqref{eq:W1.CNR.mixing}--\eqref{eq:W2.mixing}  follow from \eqref{eq:W1.mixing} essentially in the same way as assertions (B) and (C) of Lemma \ref{lem:WS1.IID} from assertion (A). The only modification required here is replacing the independence of operators $H_{\ball_L(x)}$, $H_{\ball_L(y)}$ by weak dependence at distance $O(L)$ between  $\ball_L(x)$ and $\ball_L(y)$, following from the condition \Wthree.
\qed

\begin{definition}\label{def:S}
Given a sample $V(\cdot;\om)$, a ball $ \ball_L(u)$ is called $(E,m)$-non-singular ($(E,m)$-NS) if
if for any pair of points $x,y\in \ball_L(u)$ with $|x-y|\ge L^{\frac{1+\rho}{\alpha}} = L^{7/8}$,
\be\label{eq:def.NS.IID.2}
 |\pt  \ball_L(u)|\,  
 %\mymax{y\in\pt  \ball_L(u)} 
 \; |G(x,y;E;\om)|
 \leq e^{-\gamma(m,L) |x - y| + 2L^\beta}.
\ee
Otherwise, it is called $(E,m)$-singular ($(E,m)$-S).
\end{definition}

%%%%%%%%%%%%%%%%%%%%%%%%%%%%%%%%%%%%%%%%%%%%%%%%%%%%%%%%%%%%%%%%%%%%%%%%%%%%%%%%%%%%%%%%%%%%%%%%%%%%
\subsection{Localization and tunneling}

%\begin{definition}
%Given a sample $V(\cdot;\om)$,  a ball $ \ball_L(u)$ is called $m$-localized (\loc\,, in short) if any eigenfunction $\psi_j$ of the operator $H_{ \ball_L(u)}$ operator satisfies
%\be\label{eq:def.mloc}
%|\psi_j(x) \, \psi_j(y| \le e^{-\gamma(m,L) d(x,y)}
%\ee
%for any pair of points $x,y\in  \ball_L(u)$ with
%$d(x,y) \ge L^{\frac{1+\delta+\rho}{\alpha}} = L^{\frac{1+\frac{3\delta}{2}}{1+2\delta}}$.
%\end{definition}

The definition of an $m$-localized ball remains unchanged, but we slightly modify the definition
of a tunneling ball:

\begin{definition}\label{def:tun}
A ball $ \ball_{\ell^{\alpha}}(u)$ is called $m$-tunneling if it contains a pair of  \nloc\, balls $\ball_{\ell}(v)$, $\ball_{\ell}(w)$ with $d(v,w)\le 3\ell$.
Otherwise it is called $m$-non-tunneling.
\end{definition}

%%%%%%%%%%%%%%%%%%%%%%%%%%%%%%%%%%%%%%%%%%%%%%%%%%%%%%%%%%%%%%%%%%%%%%%%%%%%%%%%%%%%%%%%%%%%%%%%%%%%
%%%%%%%%%%%%%%%%%%%%%%%%%%%%%%%%%%%%%%%%%%%%%%%%%%%%%%%%%%%%%%%%%%%%%%%%%%%%%%%%%%%%%%%%%%%%%%%%%%%%

\subsection{Initial scale bounds for correlated potentials}

%%%%%%%%%%%%%%%%%%%%%%%%%%%%%%%%%%%%%%%%%%%%%%%%%%%%%%%%%%%%%%%%%%%%%%%%%%%
\begin{lemma}[Initial scale bound: large disorder]\label{lem:L.0}
For any $L_0>2$, $m>0$ and $p>0$ there exists $g_0 < \infty$ such that for all $g$ with
$|g|\ge g_0$, any ball $\ball_{L_{0}}(u)$ and any $E\in\DR$
\begin{eqnarray}
\pr{ \ball_{L_{0}}(u) \text{ is } (E,m)\text{\rm-S} } &\le L_{0}^{-p}
\label{eq:L0.GF} \\
\pr{ \ball_{L_{0}}(u) \text{ is } \mathnloc } &\le L_{0}^{-p}.
\label{eq:L0.loc}
\end{eqnarray}
\end{lemma}

\subsection{Tunneling and localization in finite balls}

The proof of Lemma \ref{lem:loc.and.NR.imply.NS.IID} is not specific to IID potentials, so its assertion remains valid for correlated potentials. However, Lemma \ref{lem:few.nloc.and.few.R.imply.nloc.IID} needs a minor adaptation.

%%%%%%%%%%%%%%%%%%%%%%%%%%%%%%%%%%%%%%%%%%%%%%%%%%%%%%%%%%%%%%%%%%%%%%%%%%%%%%%%%%%%
\begin{lemma}\label{lem:few.nloc.and.few.R.imply.nloc}
There exists $\tilde L^{(5)}<\infty$ such that if  $L_k\ge \tilde L^{(5)}$ and
a ball $ \ball_{L_k}(u)$ is $\mathntun$ and $E$-NR, then it is also $(E,m)$-NS.
\end{lemma}
\proof

By Definition \ref{def:tun}, if a ball $ \ball_{L_k}(u)$ is $\mathntun$, then
\begin{enumerate}
  \item either it does not contain any \nloc\, ball of radius $L_{k-1}$
  \item or it contains an  \nloc\, ball $\ball_{L_{k-1}}(w)$, but any ball
  $\ball_{L_{k-1}}(v)$ with \\
  $d(v,w)>3L_k$ is $\mathloc$.
\end{enumerate}
To treat both cases with one argument, in the simpler case (1) set, formally, $w=y$.
Next, let
\be\label{eq:r.r.IID.2}
\bal
r' = \max\{ \rd(x, \ball_{L_{k-1}}(w) - 1, 0\},
%\\
%
\quad
r'' &= \max\{ \rd(y, \ball_{L_{k-1}}(w) - 1, 0\}.
\eal
\ee
Assume first that $r', r'' \ge L_{k-1}+1$. By triangle inequality,
\be\label{eq:r.r.corr}
r' + r'' \ge |x-y| - 6L_{k-1} -2  \ge L_{k-1}^{1+\rho} - 7L_{k-1}.
\ee

Consider the function  $f: \ball\times\ball \to \DR_+$ defined by
$
 f(x', x'') = |G_{\ball_{L_k}}(x', x''; E)|.
$
By construction, $f$ is  $(L_{k-1}, q_k)$-subharmonic both in
$x'\in\ball_{r'}(x)$ and in $x''\in\ball_{r''}(y)$, with
$
q \le e^{-\gamma(m, L_{k-1})},
$
since $\ball_{L_k}(u)$ is $E$-NR, and all  balls of radius $L_{k-1}$ both in $\ball_{r'}(x)$ and in $\ball_{r''}(y)$ are $(E,m)$-NS (being disjoint from $\ball_{L_{k-1}(w)}$). Collecting the assertion (B) of Lemma \ref{lem:cond.SubH.IID}, Lemma \ref{lem:BiSubH.IID} and inequality
\eqref{eq:r.r.corr}, one can write, with the convention $-\ln 0 = +\infty$:
\be\label{eq:proof.lem:few.nloc.and.few.R.imply.nloc.IID.2}
\bal
\textstyle
-\ln f(x,y)   &\ge
 -\ln \Bigg[
 \left( e^{-m(1 + \half L^{-\tau}_{k-1})L_{k-1}}\right)^{\frac{ R - 7L_{k-1} - 2 L_{k-1} }{ L_{k-1} + 1} }
 e^{L_k^\beta}
 \Bigg]
\\
&
\textstyle
\ge m \left(1+ \half L_{k-1}^{-1/8} \right) \frac{L_{k-1}}{L_{k-1}+1} \left( R -9L_{k-1} \right)
- L_{k-1}^\beta
\\
&
\textstyle
\ge R m \left( 1+ \half L_{k-1}^{-1/8} \right)  \left( 1 -  10 L_{k-1}^{-1/6} \right)
\ge R m \left(1 + \quart L_{k-1}^{-1/8}   \right)
\\
 &> \gamma(m, L_k) \,d(x, y) + \ln |\pt \ball_{L_k}(u)|,
\eal
\ee
as required for the $(E,m)$-NS property of the ball $\ball_{L_k}(u)$.

If $r'=0$ (resp., $r''=0$), the required bound follows from the subharmonicity of the function $f(x',x'')$ in $x''$ (resp., in $x'$).
\qedhere

%%%%%%%%%%%%%%%%%%%%%%%%%%%%%%%%%%%%%%%%%%%%%%%%%%%%%%%%%%%%%%%%%%%%%%%%%%%%%%%%%%%%
\begin{lemma}\label{lem:NT.implies.nloc}
There exists $\tilde L^{(6)} < \infty$ such that if  $L_k\ge \tilde L^{(6)}$ and a ball
$\ball_{L_k}(u)$ is $\mathntun$, then it is also $\mathloc$.
\end{lemma}
\proof
One can proceed as in the proof of the previous lemma, but with the functions
$
 f_j(x', x'') = \psi_j(x') \psi_j(x''),
$
where $\psi_j, j\in[1, |\ball_{L_k}(u)|]$, are normalized eigenfunctions of operator $H_{\ball_{L_k}(u)}$. Notice that the $E$-non-resonance condition for $\ball_{L_k}(u)$ is not required here, since $\|\psi_j\|_2=1$, so the function $|f_j|$ is globally bounded by  $1$.
Let $x,y\in \ball_{L_k}(u)$ and assume that $\|x-y\|=R \ge L_{k-1}^{1+\rho}$. One can apply the bound \eqref{eq:RDL.IID}, with
$
r'+r''  \ge R - 6L_{k-1} - 2,
$
and write, with the convention $-\ln 0 = +\infty$, and using the assumptions $m\ge 1$,
$\tau < \rho$:
$$
\ba
-\ln |\psi_i(x) \psi(y) |  \ge
 -\ln \left\{
\left( e^{-\gamma(m,L_{k-1})L_{k-1}}\right)^{\frac{ R - 9L_{k-1}}{ L_{k-1} + 1} }
\right\}
 \ge  \gamma(m, L_k) \,d(x, y).
 \qedhere
\ea
$$

%%%%%%%%%%%%%%%%%%%%%%%%%%%%%%%%%%%%%%%%%%%%%%%%%%%%%%%%%%%%%%%%%%%%%%%%%%%%%%%%%%%%
\begin{lemma}[\textbf{Main inductive lemma} for correlated potentials]\label{lem:induct}
Suppose that for some $k\ge 0$ the following bound holds true:
\be\label{eq:lem.induct}
\forall\, w\in\DZ^d\qquad \pr{ \text{ $\ball_{L_{k}}(w)$ is } \mathnloc } \le L_{k}^{-p(1+b)^k},
\ee
with $p, b>0$ obeying \eqref{eq:cond.p.b}.
Then for $L_0>0$ large enough and  any $u\in\DZ^d$
\be\label{eq:lem.induct.next.Loc}
\pr{ \text{  $\ball_{L_{k+1}}(u)$ is } \mathnloc } \le L_{k+1}^{-p(1+b)^{k+1}}.
\ee
\end{lemma}
\proof
%%%%%%%%%%%%%%%%%%%%%%%%%%%%%%%%%%%%%%%%%%%%%%%%%%%%%%%%%%%%%%%%%%%%%%%%%%%

Consider the following events:
$$
\bal
\CN_{k+1} &:= \{ \text{$\ball_{L_{k+1}}(w)$ is } \mathnloc \} \\
\CS_{k}^{(2)} &:= \{\exists\,E\;\exists\,  \text{ disjoint $(E,m)$-S balls } \ball_{L_{k}}(u),  \ball_{L_{k}}(v)\subset \ball_{L_{k+1}}(w) \} \\
\CR_{k}^{(2)} &:= \{\exists\,E\;\exists\,  \text{ disjoint $E$-PR balls } \ball_{L_{k}}(u),  \ball_{L_{k}}(v)\subset \ball_{L_{k+1}}(w)  \}.
\eal
$$
By Lemma \ref{lem:NT.implies.nloc.IID},
$
\CN_{k+1} \subset \CR_{k}^{(2)} \cup \CS_{k}^{(2)} \equiv
  \CR_{k}^{(2)} \cup \left(\CS_{k}^{(2)}\setminus \CR_{k}^{(2)} \right),
$
and by Wegner-type bound \eqref{eq:W2.mixing}, $\pr{\CR_k^{(2}}\le e^{- c \ln^2 L_k}$, so that it remains to bound $\pr{\CS_{k}^{(2)}\setminus \CR_{k}^{(2)} }$. Fix points $u,v$ with
$d(u, v)>2L_k$ and introduce the event
\be\label{eq:S2.u.v.2}
\CS_{k}^{(2)}(u,v) = \{\exists\,E:\;
\text{ $\ball_{L_{k}}(u)$ and $\ball_{L_{k}}(v)$  are $(E,m)$-S} \}.
\ee
Within the event $\CS_{k}^{(2)}(u,v)\setminus \CR_{k}^{(2)}$, either
$\ball_{L_k}(u)$ or $\ball_{L_k}(v)$ must be $E$-CNR; without loss of generality, suppose
$\ball_{L_k}(u)$ is $E$-CNR. At the same time, it is $(E,m)$-S, so by assertion (B) of Lemma \ref{lem:few.nloc.and.few.R.imply.nloc.IID}, $\ball_{L_k}(u)$ must  be $\mathtun$, i.e., it must contain a pair of disjoint \nloc\, balls $\ball_{L_{k-1}}(y')$, $\ball_{L_{k-1}}(y'')$. Using the inductive assumption \eqref{eq:lem.induct.IID} and
the mixing property $\Cmix$, we can write, for $L_0$ large enough,
$$
\pr{ \text{$\ball_{L_{k-1}}(y')$, $\ball_{L_{k-1}}(y')$ are \nloc} }
\le L_{k-1}^{-2p(1+b)^{k-1}  }
\le  L_{k+1}^{- \frac{2p}{\alpha^2} (1+b)^{k-1} }
$$
The number of all pairs $y', y''$ is bounded by $|\ball_{L_{k+1}}(u)|^2/2$, so that
\be\label{eq:prob.S2.not.R2.a}
\pr{ \CS_{k}^{(2)} \setminus \CR_{k}^{(2)} }
\le C_d\,  L_{k+1}^{- \frac{2p}{\alpha^2} (1+b)^{k-1} + 2d}.
\ee
Under the conditions 
\eqref{eq:cond.p.b}, the RHS is bounded by $\half L_{k+1}^{-p(1+b)^{k+1}}$.
We conclude that
$$
\pr{ \CT_{k+1} }
 \le  \pr{ \CN_{k}^{(2)} } + \pr{ \CR_{k}^{(2)} }
 \le  L_{k+1}^{-p(1+b)^{k+1}}.
 \qedhere
$$

%%%%%%%%%%%%%%%%%%%%%%%%%%%%%%%%%%%%%%%%%%%%%%%%%%%%%%%%%%%%%%%%%%%%%%%%%%%
\begin{theorem}\label{thm:ind.1p}
For any $m>0$, $p>0$ and $L_0>2$ there exist $g_0 < \infty$ and $b >0$ such that for any $g$ with $|g|\ge g_0$,  all $k\ge 0$ and any ball $\ball_{L_{k}}(u)\subset \DZ^d$,
\be\label{eq:thm.ind.1p.loc}
\pr{ \text{ $\ball_{L_{k}}(u)$ is } \mathnloc } \le L_{k}^{-p(1+ b )^k}.
\ee
In the same range of parameters, for any pair of $L_k$-distant balls $\ball_{L_{k}}(u)$,
$\ball_{L_{k}}(v)$
\be\label{eq:thm.ind.1p.NS}
\pr{ \ball_{L_{k}}(u) \text{ and } \ball_{L_{k}}(v) \text{ are } (E,m)\text{-S} }
\le L_{k}^{- p(1+ b)^k}.
\ee
\end{theorem}

\proof The first claim \eqref{eq:thm.ind.1p.loc} follows by induction from Lemma \ref{lem:induct} and Lemma \ref{lem:L.0}. The second claim can be proven in the same way as in Lemma \ref{lem:induct.IID}. \qedhere

%%%%%%%%%%%%%%%%%%%%%%%%%%%%%%%%%%%%%%%%%%%%%%%%%%%%%%%%%%%%%%%%

\section{Appendix. Proofs of auxiliary statements}

\subsection{ Proof of Lemma \ref{lem:WS1.IID} }

Assertion (A), with the RHS of the form $\Const\, e^{-{L^\beta}/b}$, is well-known; its proof can be found in many papers and reviews; cf., e.g., \cite{CL90,K07}. Assertion  (B) stems easily from (A). Indeed, the number of \emph{all} balls inside $\ball_L(u)$, with radii $0 \le r \le L$, is bounded by
$(L+1) \cdot |\ball_L(u)| = O(L^{d+1})$, so that
$$
\pr{ \text{  $\ball_{L}(x)$ is not $E$-\rm{CNR} } }
\le  \Const \, L^{d+1} e^{-L^{\beta}}
\le  e^{-L^{\beta''}}
$$
for some $0<\beta'' < \beta$, if $L$ is large enough (depending upon the value of $\beta''$). Finally, the assertion (C) can be inferred from (B) in a standard way, by conditioning on the sigma-algebra $\fF_V(\ball_L(y)$ generated by the potential inside $\ball_L(y)$, which fixes the eigenvalues $E_j(y;\om)$ of operator $H_{\ball_L(y)}(\om)$. Indeed, the LHS of \eqref{eq:W2.IID} is the probability of the event
$$
\ba
\{\exists\, E\in\DR\; \exists\, \ball_{r}(x')\subset\ball_{L}(x), \ball_{s}(y')\subset\ball_{L}(y):\;
\ball_{r'}(x'), \ball_{r''}(x'') \text{ are  $E$-R} \}
\\
\qquad = \{\exists\, E\in\DR\; \exists\, i, j:\;
  |E_i(x) - E|\le e^{L^\beta},  |E_j(y) - E|\le e^{L^\beta} \}
\\
\qquad \subset \{\exists\, i, j:\;
  |E_i(x) - E_j(y)|\le 2e^{L^\beta} \} \\
\ea
$$
and the probability of the latter event can be estimated as follows:
$$
\ba
\diy \pr{\exists\, i, j:\;
  |E_i(x) - E_j(y)|\le 2e^{L^\beta} }
\\
\diy \qquad = \esm{ \pr{ \exists\, i, j:\;
  |E_i(x) - E_j(y)|\le 2e^{L^\beta} \,\big|\, \fF_V(\ball_L(y)} }
\\
\qquad \diy \le |\ball_L(y)|\, |\ball_L(x)| \, \sup_{E\in \DR} \max_i
\pr{ |E_i(x) - E|\le 2e^{L^\beta} } \le  e^{-L^{\beta'''}},
\ea
$$
for some $0< \beta''' < \beta$, if $L$ is large enough. Now it suffices to set
$\beta' = \min(\beta'', \beta''')$.
\qed

\subsection{Proof of Lemma \ref{lem:SubH.IID}}

 Let $n\ge 1$ and consider a point $u\in  \ball_{L-n(\ell+1)}(x)$. Since
$ \ball_\ell(u) \subset  \ball_{L-(n-1)(\ell+1)}(x)\subset \ball_L(x)$, the subharmonicity condition implies that
$$
|f(u)| \le q \, \mymax{ y\in \ball_L(x):\, \|y-u\| \le \ell+1}
\le q\, \CM(f,  \ball_{L- (n-1)(\ell+1)}(x)).
$$
In other words, we have
\begin{equation}\label{eq:rad.descent.iter}
\CM( f,  \ball_{L- n(\ell+1)}(x)) \le q^{}  \CM(f,  \ball_{L-(n-1) (\ell+1)}(x)).
\end{equation}
The inequality \eqref{eq:rad.descent.iter} can be iterated,  so we obtain by induction
$$
\CM( f,  \ball_{L- n(\ell+1)}(x))
%e q^{n-1}  \CM(f,  \ball_{L-(\ell+1)}(x))
\le q^{n}  \CM(f,  \ball_{L}(x)).
$$
Now the assertion of the lemma follows from the inclusion
$
x \in  \ball_{L - \left[\frac{L}{\ell+1}\right](\ell+1)}(x).
$
\qed

\subsection{Proof of Lemma \ref{lem:BiSubH.IID}}

Fix any point $y\in \ball_{r''}(u'')$. Then the function $g(x) := f(x,y)$ is $(\ell,q)$-subharmonic in $x\in \ball_{r'}(u')$. Therefore,
$$
|g(x)| \leq q^{\left[\frac{r'+1}{\ell+1}\right]} \CM(f, \ball)
$$
and, since $y\in \ball_{r''}(u'')$ is arbitrary,
$$
\max_{y\in \ball''_{r''}(u'')} \, |f(x',y)| \le q^{\left[\frac{r'+1}{\ell+1}\right]}
\CM(f, \ball).
$$
Next, the function $h(y) := f(x',y)$  is $(\ell,q)$-subharmonic in $\ball_{r''}(u'')$, with global maximum
$
\CM(h, \ball_{r''}(u'')) \le q^{\left[\frac{r'+1}{\ell+1}\right]} \CM(f, \ball),
$
and the subharmonicity of the function $h$ gives the desired upper bound:
$$
\ba
\qquad
|f(u', u'')| = |h(u'')| \leq q^{\left[\frac{r''+1}{\ell+1}\right]} \CM(h, \ball_{r''}(u''))
\tabhigh{ \le q^{\left[\frac{r'+1}{\ell+1}\right]+\left[\frac{r''+1}{\ell+1}\right]}
\CM(f, \ball).}
\qquad\qed
\ea
$$

\subsection{Proof of Lemma \ref{lem:cond.SubH.IID} }
$\,$

\par
\noindent
(A) The  bound $|\psi_j(v)|\le 1$ follows from the normalization condition $\|\psi_j\|_2=1$. It suffices to prove the $(\ell,q)$-subharmonicity of the functions $x\mapsto |\psi_j(x)|$  (with $q$ given by \eqref{eq:lem:ConSubH.IID}), for the kernel of the eigenprojection $\Pi_j$ has the form
$\Pi_j(x, y) = \psi_j(x) \, \overline{\psi_j(y)}$. Since all balls
$\ball_\ell(v)$ with $v\in\ball_R(x')$ are assumed to be $(E,m)$-NS, a direct application of the GRI for the eigenfunction (cf. Eqn\eqref{eq:GRI.EF}) gives
\be
\bal
| \psi(v) |
& \le \left(\, C_\ell \,  \max_{w: d(v,w)=\ell} |G_{\ball_\ell(v)}(v,w;E) | \,\right)
\; \max_{z: d(v,z)\le \ell+1} \, | \psi(z) |, \\
& \le
%C(d) \ell^{d-1}
e^{-\gamma(m,\ell)\ell}
\max_{z: d(v,z)\le \ell+1} \, | \psi(z) |
%
%&
\le q \, \max_{z: d(v,z)\le \ell+1} \, | \psi(z) |.
\eal
\ee

\par
\noindent
(B) Since the ball $\ball_L(u)$ is $E$-NR, we have
$\|G_{\ball_L(u)(E)}\|\le e^{L^\beta}$. Now the $(\ell,q)$-subharmonicity of the Green functions  follows from the assumption of non-singularity of all boxes $\ball_\ell(v)$ with $v\in\ball_R(x')$ by a direct application of the GRI.
\qed

\subsection{Proof of Lemma \ref{lem:L.0.IID} }
If $F_V$ is  continuous, then for any positive number $\delta$, including
$\delta = L_0^{-p}$, and  for $\epsilon_0$ small enough we have
$$
\sup_{E\in\DR} \; \pr{ \exists\, x\in \ball_{L_0}(u):\, |V(x;\omega) - E| \le \epsilon_0} \le \delta \;( = L_0^{-p}).
$$
Fix an arbitrary $E\in\DR$ and assume that
$
\omega \in \Om_\epsilon:=\left\{ \forall\, x\in \ball_{L_0}(u) |V(x;\omega) - E| \ge \epsilon \right\}.
$
Operator $V -g^{-1}E$ is diagonal, and all its eigenvectors have the form $\delta_x(y) = \delta_{x,y}$. Observe that $\|H_0\|<\infty$ and, for $\om\in\Om_\epsilon$, we have
$$
\min_{x\in \ball_{L_0} } \; |gV(x;\omega) - E| \ge g\epsilon_0 \tto{|g|\to \infty}{}{} \infty.
$$
Write now $H-E= g(V -g^{-1}E + g^{-1} H_0)$. The property \eqref{eq:cond.continuity.V} implies that with probability one all eigenvalues $E_j(\om)$ of the operator $V_{\ball_L(u)}(\om)$ are distinct, and all spacings
$|E_j(\om) - E_i(\om)|$ are positive. For $|g|$ large enough, all spacings for operator $gV(\om)$ are arbitrarily large. Eigenvectors of a  continuous operator family $A(t)$ with simple spectrum at $t=t_0$ are continuous in a neighborhood of $t_0$. For the second assertion, it suffices to apply this fact to the family $A(t) = V + g^{-1}t\Delta$, $t\in[0,1]$.
\qed

\subsection{Proof of Lemma \ref{lem:L.0}}

Since the CDF $F_{V,x}$ are (uniformly) continuous, we  have
$$
\bal
\sup_{E\in\DR} \; \pr{ \exists\, x\in \ball_{L_0}(u):\, |V(x;\omega) - E| \le \epsilon}
&\tto{ \epsilon\to 0}{0}{} \\
\pr{ \exists\, x,y\in \ball_{L_0}(u), x\ne y:\, |V(x;\omega) - V(y;\omega) | \le \epsilon}
&\tto{ \epsilon\to 0}{0}{}.
\eal
$$
Now one can conclude as in the proof of Lemma \ref{lem:L.0.IID}.
\qed

\section*{Acknowledgements.}

It is a pleasure to thank Tom Spencer, Boris Shapiro, Abel Klein and Misha Goldstein for stimulating and fruitful discussions of localization techniques; the organizers of the program
\textit{"Mathematics and Physics of Anderson Localization: 50 Years After"} at the Isaac Newton Institute, Cambridge (2008);  Shmuel Fishman, Boris Shapiro and the Department of Physics of Technion, Israel, for their warm hospitality.

%---------------------------------------------------------------------%
%:bib
%---------------------------------------------------------------------%

%---------------------------------------------------------------------%
%:end
%---------------------------------------------------------------------%
\end{document}